\documentclass[3p,times,review]{elsarticle}

\journal{Journal of Computer and System Sciences}

\usepackage{amssymb}

\usepackage[figuresright]{rotating}

\newtheorem{theorem}{Theorem}
\newtheorem{lemma}[theorem]{Lemma}
\newtheorem{proposition}[theorem]{Proposition}
\newtheorem{corollary}[theorem]{Corollary}
\newtheorem{definition}[theorem]{Definition}
\newtheorem{claim}{Claim}
\newtheorem{example}{Example}

\newproof{proof}{Proof}

\usepackage[utf8]{inputenc}
\usepackage{complexity}
\usepackage[algosection]{algorithm2e}
\usepackage{endnotes}
\usepackage{listings}
\usepackage{mathtools}
\usepackage{verbatim}
\usepackage{mathrsfs}
\usepackage{shuffle}
\usepackage{longtable}
\usepackage{enumitem}
\usepackage{etoolbox}
\usepackage{float}
\usepackage{color}
\usepackage{hyperref}

\usepackage{mathtools}

\newcommand*{\DEBUG}{}%

\ifdefined\DEBUG
\newcommand{\fixme}[1]{{\textcolor{red}{\bf{\textsf{FIXME: #1}}}}}
\newcommand{\bug}[1]{{\textcolor{blue}{\bf{\textsf{BUG: #1}}}}}
\newcommand{\idea}[1]{{\textcolor{blue}{\bf{\textsf{IDEA: #1}}}}}

\newcommand{\TODO}[1]{{\textcolor{red}{\bf{\textsf{
TODO: #1
}}}}}
\else
\newcommand{\fixme}[1]{}
\newcommand{\bug}[1]{}
\newcommand{\TODO}[1]{}
\newcommand{\idea}[1]{}
\fi

\newclass{\DCM}{DCM}
\newclass{\RCM}{RCM}
\newclass{\DCA}{DCA}
\newclass{\DPDA}{DPDA}
\newclass{\DCMNE}{DCM_{NE}}
\newclass{\TwoDCM}{2DCM}
\newclass{\NCM}{NCM}
\newclass{\DPCM}{DPCM}
\newclass{\NPCM}{NPCM}
\newclass{\TRE}{TRE}

\newcommand\abs[1]{\left|#1\right|}

\newcommand\set[1]{\left\{#1\right\}}

\newcommand{\sst}{\ensuremath{\mid}}

\newcommand\union{\cup}

\newcommand\natnum{\mathbb{N}}
\newcommand\natzero{\mathbb{N}_0}

\DeclareMathOperator{\pref}{pref}
\DeclareMathOperator{\suff}{suff}
\DeclareMathOperator{\infx}{inf}
\DeclareMathOperator{\outf}{outf}
\DeclareMathOperator{\emb}{emb}

\begin{document}

\begin{frontmatter}




\title{Insertion Operations on Deterministic Reversal-Bounded Counter Machines\tnoteref{t1}}

\tnotetext[t1]{\textcopyright 2018. This manuscript version is made available under the CC-BY-NC-ND 4.0 license \url{http://creativecommons.org/licenses/by-nc-nd/4.0/}}


\author[label1]{Joey Eremondi\fnref{fn3}}
\address[label1]{Department of Computer Science\\ University of British Columbia, Vancouver, BC V6T 1Z4, Canada}
\ead[label1]{jeremond@cs.ubc.ca}

\author[label2]{Oscar H. Ibarra\fnref{fn2}}
\address[label2]{Department of Computer Science\\ University of California, Santa Barbara, CA 93106, USA}
\ead[label2]{ibarra@cs.ucsb.edu}
\fntext[fn2]{Supported, in part, by
NSF Grant CCF-1117708 (Oscar H. Ibarra).}

\author[label3]{Ian McQuillan\fnref{fn3}}
\address[label3]{Department of Computer Science, University of Saskatchewan\\
Saskatoon, SK S7N 5A9, Canada}
\ead[label3]{mcquillan@cs.usask.ca}
\fntext[fn3]{Supported, in part, by Natural Sciences and Engineering Research Council of Canada Grant 2016-06172 (Ian McQuillan).}

\begin{abstract}
Several insertion operations are studied applied to languages accepted by one-way and two-way deterministic reversal-bounded multicounter machines. These operations are defined by the ideals obtained from relations such as the prefix, infix, suffix, and outfix relations, as well as operations defined from inverses of a type of deterministic transducer with reversal-bounded counters attached. The question of whether the resulting languages can always be accepted by deterministic machines with the same number (or larger number)
of input-turns (resp., counters, counter-reversals, etc.) is investigated.
\end{abstract}


\end{frontmatter}





\section{Introduction}

One-way deterministic multicounter machines are deterministic finite automata augmented by a fixed number of counters, which can each be independently increased, decreased or tested for zero. If there is a bound on the number of switches each counter makes between increasing and decreasing, then the machine is reversal-bounded
\cite{Baker1974,Ibarra1978}. The family of languages accepted by one-way deterministic reversal-bounded multicounter machines is denoted by $\DCM$ (and $\DCM(k,l)$ when there are at most $k$ counters with at most an $l$-reversal-bound), and the nondeterministic variant is denoted by $\NCM$. 

Reversal-bounded counter machines (both deterministic and nondeterministic)
have been extensively studied.  Many generalizations have been investigated,
and they have found applications in areas such as verification of infinite-state systems \cite{verification,stringTransducers,modelChecking,verificationDiophantine}, membrane computing systems \cite{counterMembrane}, biocomputing \cite{bioOp}, Diophantine equations \cite{verificationDiophantine}, and others.
$\DCM$ in particular is an interesting family as it is more general than the family of regular languages, but still has  decidable  emptiness, infiniteness, equivalence, inclusion, universe, and disjointness problems
\cite{Ibarra1978}. Moreover, these problems remain decidable if the machines operate with two-way input that is
finite-crossing in the sense that there is a fixed $r$ such that
the number of times the boundary between any two adjacent input
cells is crossed is at most $r$ \cite{Gurari1981220}. In addition,
for fixed $k,l$, the emptiness, membership, containment, and equivalence problems for $\DCM(k,l)$ can be
tested in polynomial time \cite{Gurari1981220}. 
Hence, $\DCM$ has many nice
decidability and complexity theoretic properties. 
We know of no other family more general than the regular languages
that enjoy these properties.
Despite this,
little is known regarding the closure properties of this family, which
is important for constructing other languages that remain in this family.

More recently, the $\DCM$ model has gained a resurgence of theoretical interest. It was shown that all commutative semilinear languages
are in $\DCM$, and in fact, the subfamily of $\DCM$ languages accepted
by machines
that cannot subtract from any counter until hitting the right end-marker was shown to be equal to the smallest family closed under inverse
deterministic finite transductions, commutative closure, and right quotient with regular languages \cite{eDCM}. In \cite{shuffleInformationAndComputation}, it
was shown that there is a polynomial time algorithm to decide,
for fixed $k,l$ whether the shuffle of two
$\NCM(k,l)$ machines is contained in a $\DCM(k,l)$ machine.
In addition, $\DCM$ was studied in \cite{Castiglione2016} as part
of an interesting conjecture involving holonomic functions.
The authors define a family $\RCM$ that is obtained from the regular languages
via so-called linear constraints on the number of occurrences of symbols,
and homomorphisms. It is demonstrated that all $\RCM$ languages
have generating functions which are all holonomic functions. The
class of holonomic functions in one variable is an
extension of the algebraic functions which contains all those
functions satisfying a linear differential equation with
polynomial coefficients \cite{Castiglione2016}. They conjectured
that $\DCM$ is contained in $\RCM$, implying that all $\DCM$
languages have holonomic generating functions. Although this conjecture
has yet to be established, it is shown to be true for a subfamily of
$\DCM$.
The study of closure properties on $\DCM$ can potentially help in this regard towards establishing the conjecture.
Deletion operations applied to $\DCM$ have also been recently investigated \cite{deletionDCM} using word operations such as prefix, suffix, infix, and left and right quotients. It was found that $\DCM$ is closed under right quotient with many general families defined even by nondeterministic machines such as the context-free languages, and it is shown that the left quotient of a $\DCM(1,1)$ language with general families such as the context-free languages always gives $\DCM$ languages. However, even the suffix closure of languages in $\DCM(2,1)$ or $\DCM(1,3)$ gives languages which are not in $\DCM$.

Generally, various schema for insertions and deletions have been studied in automata theory, from simple concatenation \cite{HU}, to more complex insertion operations \cite{parInsDel}, and they have found applications in the area of natural computing for modelling biological processes \cite{inter,TranspositionsNatCo}. 

In this paper, we study various insertion operations on deterministic reversal-bounded multicounter languages. 
The prefix, suffix, infix, and outfix deletion operations can
also be used to define insertion operations.
As an example, the set of all infixes of a language $L$,
$\infx(L) = \{w \mid xwy \in L, x,y \in \Sigma^*\}$, and then 
the inverse of this operation, $\infx^{-1}(L) = \Sigma^* L \Sigma^*$, is the set of all words having a word in $L$
as an infix. This is the same as what is often called the {\em two-sided ideal}, or the {\em infix ideal} \cite{JKT}. For the suffix operation, $\suff(L) = \{ w \mid xw \in L, x \in \Sigma^*\}$, and $\suff^{-1}(L) = \Sigma^* L$, with the latter being called the {\em left ideal}, or the {\em suffix ideal}. For prefix, $\pref(L) = \{ w \mid wy \in L, y \in \Sigma^*\}$, and $\pref^{-1}(L) = L \Sigma^*$, the {\em prefix ideal}, or the {\em right ideal}.
Thus, the inverse of each operation defines a natural and simple insertion operation.

We will examine the insertion operations defined by the inverse of the prefix, suffix, infix, outfix, and embedding operations, as well
as the concatenation of languages from other families.
It is easy to see that all language families 
closed under homomorphism, inverse homomorphism and
intersection with regular languages (such as the 
nondeterministic reversal-bounded
multicounter languages, or the context-free languages) are closed
under all these insertion operations. However, this is a more
complex question for families accepted by deterministic machines such as $\DCM$.
In this case, if we start with a language that can be accepted with a parameterized number of counters, input tape turns, and reversals on the counters, is the result of the various insertion operations always accepted with the same type of machine? These results are summarized for all such insertion operations in column 2 of Table \ref{tab:summary}. And if they are not closed, can they always be accepted by increasing either the number of counters, or reversals on the counters (presented in column 3 of Table \ref{tab:summary}),
or turns on the input tape (listed in Section \ref{sec:summary})?
Results in this paper form a complete characterization in this regard. In particular, it is rather surprising that even if we have languages accepted by deterministic 1-reversal-bounded machines with either one-way input and $2$ counters, or $1$ counter and $1$ turn on the input, then concatenating $\Sigma^*$ to the right can result in languages that can neither be accepted by $\DCM$ machines (any number of reversal-bounded counters), nor by two-way deterministic reversal-bounded one counter machines ($2\DCM(1)$, which have no bound on input turns). This is in contrast to deterministic pushdown languages which are closed under right concatenation with regular languages \cite{harrison1978}. In addition, concatenating $\Sigma^*$ to the left of a $\DCM(1,1)$ language can create languages that are neither in $\DCM$ nor $2\DCM(1)$.

\begin{table}[t]
\centering
\begin{tabular}{|l|ll|ll|}
\hline
\textbf{Operation}
& is $ Op(L) \in \DCM(k,l)$? &  & is $Op(L) \in \DCM$? &
 \\\hline$\pref^{-1}(L)$  &Yes if $k = 1, l \geq 1$ & Cor \ref{dcminverseprefix} & Yes if $k = 1, l \geq 1$ & Cor \ref{dcminverseprefix}
  \\&No if $k\geq 2, l \geq 1$ & Thm \ref{nonclosure2}  & Yes if $L \in \DCMNE$ & Thm \ref{DCMNEwithREG}
  \\&                            & & No otherwise if $k \geq 2, l \geq 1$ & Thm \ref{nonclosure2}
 \\\hline$\suff^{-1}(L)$  &No if $k,l\geq 1$ & Thm \ref{inversesuffix} &  No if $k,l \geq 1$ & Thm \ref{inversesuffix}
 \\\hline$\infx^{-1}(L)$  &No if $k,l\geq 1$ & Thm \ref{prop:inverseNotInDCM} &  No if $k,l \geq 1$ & Thm \ref{prop:inverseNotInDCM}
 \\\hline$\outf^{-1}(L)$   &No if $k,l\geq 1$ & Thm \ref{inverseoutfix} & No if $k,l \geq 1$ & Thm \ref{inverseoutfix}

 \\\hline$\emb^{-1}(m,L)$   &No if $k,l,m\geq 1$ & Cor \ref{inverseembedding} & No if $k,l,m \geq 1$ & Cor \ref{inverseembedding}

 \\\hline$LR$  &Yes if $k = 1, l \geq 1$ & Cor \ref{oneCounterConcat} & Yes if $k = 1, l \geq 1$  & Cor \ref{oneCounterConcat}
 \\ &  Yes if $L \in \DCMNE$ & Thm \ref{DCMNEwithREG} & Yes if $L \in \DCMNE$ & Thm \ref{DCMNEwithREG}
    \\   &No otherwise if $k \geq 2, l \geq 1$ & Thm \ref{nonclosure2} &  No otherwise if $k \geq 2, l \geq 1$ & Thm \ref{nonclosure2}

 \\\hline$R L$  & Yes if $R$ prefix-free & Cor \ref{prefixfreeregular} &  Yes if $R$ prefix-free & Cor \ref{prefixfreeregular}
   \\   &No otherwise if $k,l \geq 1$ & Cor \ref{leftreg} & No otherwise if $k,l \geq 1$ & Cor \ref{leftreg}
 \\\hline$L_\DCM L$  &No if $k,l \geq 1$ & Cor \ref{leftdcmne} & No if $k,l \geq 1$ & Cor \ref{leftdcmne}

 \\\hline$L_\DCMNE L$  &No if $k,l \geq 1$ & Cor \ref{leftdcmne} & Yes if $L_\DCMNE$ prefix-free & Thm \ref{concatenationprefixfreeDCM}
 \\                                                             & & & No otherwise if $k,l \geq 1$ & Cor \ref{leftdcmne}


\\\hline
\end{tabular}
\caption{Summary of results for $\DCM$. Assume $R \in \REG$, $L_\DCM \in \DCM$, and $L_\DCMNE \in \DCMNE$.
Then, for all $L \in \DCM(k,l)$, the question 
in row 1 is presented
for each insertion operation in column 1. When applying the operation in the first column to any $L\in \DCM(k,l)$, is the result necessarily in $\DCM(k,l)$ (column 2), and in $\DCM$ (column 3)? This is parameterized in terms of $k$ and $l$, and the theorems showing each result is provided.}
\label{tab:summary}
\end{table}

As a consequence of the results in this paper, it is evident that the right input end-marker used for language acceptance for (one-way) $\DCM$ strictly increases the power for even one-way deterministic reversal-bounded multicounter languages when there are at least two counters. This is usually not the case for various classes of one-way machines, such
as for deterministic pushdown automata ($\DPDA$s). 
Indeed, language acceptance for $\DPDA$s is defined as being without a
right end-marker, and $\DPDA$s are closed under right quotient
with a single symbol \cite{GinsburgDPDAs}, meaning a right end-marker could be
removed without altering the languages accepted. In contrast, language
acceptance for $\DCM$s is defined using a right end-marker and
$\DCM$ is closed under right quotient with symbols (and even context-free
languages) \cite{deletionDCM}. But the end-marker is necessary for this right quotient
result. Moreover, if language acceptance for $\DCM$ is defined
without an end-marker (defined and studied in this paper), this family of languages is not closed under right quotient with a single symbol. This demonstrates
the importance of the right input end-marker.

Lastly, a type of finite transducer augmented by reversal-bounded counters
is studied, and it is shown that $\DCM$ is closed under these inverse
deterministic transductions.
The inverses of these transductions can be used for defining many
insertion operations under which $\DCM$ is closed.

Most non-closure results in this paper use techniques that
simultaneously shows languages are not in $\DCM$ and not
in $2\DCM(1)$.  The techniques do not rely on any pumping arguments.

\section{Preliminaries}

The set of non-negative integers is represented by $\natzero$, and positive integers by $\natnum$. For $c \in \natzero$,
let $\pi(c)$ be $0$ if $c=0$, and $1$ otherwise.

We use standard notations for formal languages,
referring the reader to \cite{harrison1978,HU}.
The empty word is denoted by $\lambda$.
We use $\Sigma$ and $\Gamma$ to represent finite alphabets,
with $\Sigma^*$ as the set of all words over $\Sigma$
and $\Sigma^+ = \Sigma^* \setminus \set{\lambda}$.
For a word $w \in \Sigma^*$, if $w = a_1 \cdots a_n$
where $a_i \in \Sigma$, $1\leq i \leq n$, the length of $w$ is denoted by $\abs{w}=n$,
and the reversal of $w$ is denoted by $w^R = a_n \cdots a_1$. The number of $a$'s, for $a\in \Sigma$, in $w$ is $|w|_a$. Given a language $L\subseteq \Sigma^*$,
the complement of $L$, $\Sigma^* \setminus L$ is denoted by $\overline{L}$.


\begin{definition}
\label{def:opGeneralize}
For a language $L \subseteq \Sigma^*$, we define
the prefix, inverse prefix, suffix, inverse suffix, infix, inverse infix, outfix and inverse outfix operations, respectively:
\begin{center}
$\begin{array}{|l|l|}\hline
\ \pref(L) = \set{w \sst wx \in L, x \in \Sigma^* }	&\ \pref^{-1}(L)  = \set{wx \sst w \in L, x \in \Sigma^* }\\
\ \suff(L) = \set{w \sst xw \in L, x \in \Sigma^* }	&\ \suff^{-1}(L) = \set{xw \sst w \in L, x \in \Sigma^* }\\
\ \infx(L) = \set{w \sst xwy \in L, x,y \in \Sigma^* }	&\ \inf^{-1}(L) = \set{xwy \sst w \in L, x,y \in \Sigma^* }\\
\ \outf(L) = \set{xy \sst xwy \in L, w \in \Sigma^* } 	&\ \outf^{-1}(L) = \set{xwy \sst xy \in L, w \in \Sigma^* }
\\\hline
\end{array}$
\end{center}
\end{definition}

We generalize the outfix relation to the notion of embedding \cite{JKT}
(introduced in \cite{longChinese}):
\begin{definition} 
The $m$-embedding of a language $L \subseteq \Sigma^*$ is 
$\emb(L, m) = \{w_0 \cdots w_m \sst w_0 x_1 \cdots w_{m-1} x_m w_m \in L$,
$w_i \in \Sigma^*, 0 \leq i \leq m, x_j \in \Sigma^*, 1 \leq j \leq m\}$.
We define the inverse as follows:
$\emb^{-1}(L, m) = \{w_0 x_1 \cdots w_{m-1} x_m w_m \sst  w_0 \cdots w_m \in L$,
$w_i \in \Sigma^*, 0 \leq i \leq m, x_j \in \Sigma^*, 1 \leq j \leq m$
$\}$.
\end{definition}
Note that $\outf(L) = \emb(L, 1)$ and $\outf^{-1}(L) = \emb^{-1}(L, 1)$.

We emphasize again that although these operations are defined
via inverses, most are very simple when viewed as insertion operations
with $\pref^{-1}(L) = L \Sigma^*, \suff^{-1}(L) = \Sigma^* L$,
and $\inf^{-1}(L) = \Sigma^* L \Sigma^*$. The inverse outfix
operation inserts a word at an arbitrary position of every word, where
the inverse $m$-embedding inserts $m$ arbitrary words in every word.

A language $L$ is called \textit{prefix-free} if, for all words $x,y \in L$, where $x$ is a prefix of $y$, then $x=y$.

A {\em one-way $k$-counter machine} is a tuple $M = (k,Q,\Sigma, \lhd, \delta, q_0, F)$, where
$Q, \Sigma, \lhd, q_0,F$ are respectively the finite set of states, the input alphabet, the right end-marker (not in $\Sigma$), the initial state in $Q$, and the set of final states, which is a subset of $Q$. The transition function $\delta$ (defined as in \cite{Ibarra1978} except with only a right end-marker since these machines only use one-way inputs) is a partial function from $Q \times (\Sigma \cup \{\lhd\}) \times \{0,1\}^k$ into the family of subsets of $Q \times \{{\rm S},{\rm R}\} \times  \{-1, 0, +1\}^k$, such that if $\delta(q,a,c_1, \ldots, c_k)$ contains $(p,d,d_1, \ldots, d_k)$ and $c_i =0$ for some $i$, then $d_i \geq 0$ to prevent negative values in any counter.
The symbols ${\rm S}$ and ${\rm R}$ indicate the direction of input tape head movement, either {\em stay} or {\em right} respectively. The machine $M$
is {\em deterministic} if every element mapped by $\delta$ is to a subset
of size one. The machine $M$ is {\em non-exiting} if there are no transitions defined on final states.
A {\em configuration} of $M$ is a $k+2$-tuple $(q, w , c_1, \ldots, c_k)$ representing the fact
that $M$ is in state $q$, $w$ (either in $\Sigma^*$ or $\Sigma^* \lhd$) is the remaining input, and $c_1, \ldots, c_k\in \natzero$ are the
contents of the $k$ counters. The relation $\vdash_M$ is defined between configurations, where $(q, aw, c_1, \ldots , c_k) \vdash_M (p, w'$ $, c_1 + d_1, \ldots, c_k+d_k)$, if
$(p, d, d_1, \ldots, d_k) \in \delta(q, a, \pi(c_1), \ldots, \pi(c_k))$ where $d \in \{{\rm S}, {\rm R}\}$ and
$w' =aw$ if $d={\rm S}$, and $w' = w$ if $d={\rm R}$. We let $\vdash^*_M$ be the reflexive, transitive closure of $\vdash_M$. And, for $m\in \natzero$, let $\vdash^m_M$ be the application of $\vdash_M$ $m$ times.
A word $w\in \Sigma^*$ is accepted by $M$ if
$(q_0, w\lhd, 0, \ldots, 0) \vdash_M^* (q, \lhd, c_1, \ldots, c_k)$, for some $q \in F$, and
$c_1, \ldots, c_k \in \natzero$. A derivation between the initial configuration and a final configuration is called an {\em accepting computation}. 
The language accepted by (final state in) $M$, denoted by $L(M)$,
is the set of all words accepted by $M$.

The machine $M$ is $l$-reversal-bounded if, in every accepting computation, the count on each counter alternates between increasing and decreasing at most $l$ times.

We denote by $\NCM(k,l)$ the family of languages accepted by one-way
nondeterministic $l$-reversal-bounded $k$-counter machines.
We denote by $\DCM(k,l)$ the family of languages accepted by
one-way deterministic $l$-reversal-bounded $k$-counter machines.
The union of the families of languages are denoted by
 $\NCM = \bigcup_{k,l \geq 0} \NCM(k,l)$ and
 $\DCM = \bigcup_{k,l \geq 0} \DCM(k,l)$.

Given a $\DCM$ machine $M =(k,Q,\Sigma, \lhd, \delta, q_0, F)$, the language accepted by \textit{final state without end-marker} is the set of words $w$ such that $(q_0, w\lhd, 0, \ldots, 0) \vdash_M^* (q', a\lhd, c'_1, \ldots, c'_k ) \vdash_M (q, \lhd, c_1, \ldots, c_k)$, for some $q \in F$, $q' \in Q$, $a \in \Sigma$, $c_i, c'_i \in \natzero, 1 \leq i \leq k$.
Such a machine does not ``know" when it has reached the end-marker $\lhd$. The state that the machine is in when the last letter of input from $\Sigma$ is consumed
entirely determines acceptance or rejection. It would be equivalent to require
$(q_0, w, 0, \ldots, 0) \vdash_M^* (q, \lambda, c_1, \ldots, c_k), w \in \Sigma^*$, for some $q \in F$, but we continue to use $\lhd$ for compatibility with the end-marker definition.
We use $\DCMNE(k,l)$ to denote the family of languages accepted
by these machines by final state without end-marker when they have $k$ counters that are $l$-reversal-bounded.
We define $\DCMNE = \bigcup_{k,l \geq 0} \DCMNE(k,l)$.


We denote by $\TwoDCM(1)$ the family of languages accepted by two-way deterministic finite machines (with both a left and right input tape end-marker) augmented by one reversal-bounded counter, accepted by final state. A machine of this form is said to be {\em finite-crossing} if there 
is a fixed $k$ such that
the number of times the boundary between any two adjacent input
cells is crossed is at most $k$ times, and a machine of this form is {\em finite-turn bounded}
if there is a fixed $t$ where $M$ makes at most $t$ changes of direction on the input tape for every computation
\cite{Gurari1981220}. Note a finite-turn machine is finite-crossing, but the converse does not hold in general.
The family $\NPCM$ ($\DPCM$) is defined by languages accepted by
one-way nondeterministic (deterministic) machines with an unrestricted
pushdown augmented by reversal-bounded counters \cite{Ibarra1978}.


\section{Closure for Insertion and Concatenation Operations}


Closure under concatenation is difficult for $\DCM$ languages because of determinism. However, certain special cases are demonstrated where
closure can be obtained. Towards this, a comparison of 
$\DCM$ to $\DCMNE$ will be made. This is important as it will
be shown that $\DCMNE$ is closed under right concatenation with
regular languages, although this will be shown not to be true generally for $\DCM$. However, when only one reversal-bounded counter is used,
the end-marker will be shown to not change the capacity. This will show
that $\DCM$ languages defined by machines with one reversal-bounded counter are closed under right
concatenation with regular languages. In addition, closure under left concatenation with prefix-free regular languages will be shown.
These results serve to demonstrate that $\DCM$ languages are strictly more powerful with the end-marker, but add no power
to $\DCM(1,l)$. This is in contrast to deterministic pushdown automata which do not need a right input end-marker.

\begin{lemma}
\label{lem:noEndMarker}
For any $l \geq 1$, $\DCM(1,l) = \DCMNE(1,l)$.
\end{lemma}
\begin{proof}
Trivially, $\DCMNE(1,l) \subseteq \DCM(1,l)$, by removing all transitions defined on the end-marker.

For the reverse containment, consider $M=(1,Q,\Sigma,\lhd,\delta,q_0,F)$ accepting $L$ by final state. A machine $M'$ will be built such that the language accepted by $M'$ by final state without end-marker is equal to $L(M)$.

We assume without loss of generality that $\delta$ is a total function.
Let $|Q| = n$.
For each state $q \in Q$, define the language
$$L(q)= \{a^i \mid (q,\lhd,i) \vdash_M^* (q_f,\lhd,c), q_f \in F, c\in \natzero \},$$ the set of counter values which lead
to acceptance from the end-marker $\lhd$ and state $q$.
This language can be accepted by a machine in $\DCM$ with one counter, by adding the input $i$ to the counter, then simulating $M$ from state $q$, and accepting if $M$ does.
Since all $\DCM$ languages are semilinear \cite{Ibarra1978}, $L(q)$ is unary, all unary semilinear languages are regular \cite{harrison1978}, then $L(q)$ is regular.
Thus we can accept $L(q)$ with a DFA, say $D(q)= (Q_{D(q)},\{a\},\delta_{D(q)},s_{D(q)},F_{D(q)})$.

Because these languages are unary, the structure of the DFAs are relatively simple, and well-known (see \cite{Chrobak} for
a seminal work on unary finite automata, and \cite{unaryReg} for
the informal language used here).
Every unary DFA with $m$ states is isomorphic to one with states $\{0, \ldots, m-1\}$
where there exists some state $t$, and there is a transition from $i$ to $i+1$, for all $0 \leq i <t$ (the ``tail''), and there is a transition from $j$ to $j+1$ for all  $t \leq j <m-2$, plus a transition from $m-1$ to $t$ (the ``loop''), and no other transitions.
Furthermore, we can assume without loss of
generality that each $D(q)$, for $q\in Q$, has the same length of tail,
$t>0$, equal to the maximum of the tail lengths of all of the original
DFAs $D(p)$, over all $p \in Q$. This can be done as the tail of a unary DFA can be made longer by adding additional states to the tail and shifting the final states in the loop.
Similarly, it can be assumed that the loops are all of the same length
$l$ by making it the length that is the least common multiple
of the original lengths (thus making the loop
length a multiple of all the originals). Thus all $D(q)$, for $q \in Q$
have tail length $t>1$, loop length $l$, and $m$ states, and they all
differ only in final states. Let $\delta_D(i)$ be the state
of all $D(q)$ machines after reading $a^i$, $i \geq 0$. Then,
$\delta_D(i)$ is $i$ if $i \leq t$, and $i - t \mod l$ otherwise.

The intuition for the construction of $M'$ is as follows. The machine $M'$ simulates $M$, and after reading $w$, if $M$ has counter value $c$,
$M'$ has counter value $c-t$ if $c> t$, with $t$ stored in the finite control. If $c \leq t$, then
$M'$ stores $c$ in the finite control with zero on the counter.
This allows $M'$ to know what counter value $M$ would have after reading a given word,
but also to know when the counter value is less than $t$ (and the specific value less than $t$).
In the finite control, in addition to simulating $M$, $M'$ simulates each $D(q)$, for all $q \in Q$, in parallel in such a way that the (unique, for all DFAs) state $\delta(i)$ is stored when the counter of $M$ is $i$. To do this, $M'$ stores two integers, $(d,j)$, where
$0 \leq d \leq t, 0 \leq j \leq l$, and if $i<t$, then 
$(d,j) = (i,0)$, and if $i \geq t$, then $(d,j) = (t, i-t \mod l)$.
Thus, we call the first component the ``tail'' counter, and the second
component the ``loop'' counter.
Then $d+j$ is the state of $D(q)$ after reading $a^i$.
Each time $M$ increases the counter, from $i$ to $i+1$, the state of each $D(q)$ is determined by increasing the appropriate bounded
counter (the first component of it is not yet $t$, and the second component otherwise). Each time $M$ decreases the counter from $i$ to $i-1$, the state of each $D(q)$ changes deterministically by decreasing the second component $j$ by one modulo $l$ if $i>t$ (going ``backwards'' in the loop), and if $i \leq t$, then the counter of $M'$ will be zero, and thus the simulation of each $D(q)$ can tell when to switch deterministically from decreasing the loop counter to the tail counter.
Then, when $M$ is in state $q$, $M'$ can tell if the current
counter value would be accepted using the appropriate DFA $D(q)$.

We now provide the construction in detail:

The machine $M'$ has state set
$Q_{M'} = (Q \times \set{0, \ldots, t} \times \{0\}) \cup (Q \times \{t\} \times \{0, \ldots, l\})$.
The final states of $M'$ are of the form 
$(q, d, j)$ if
either ($d<t, j=0, d \in F_{D(q)}$), or ($d = t, t+j \in F_{D(q)}$).

If $\delta(q,b, 0) = (p,T,\alpha)$, for some $q,p \in Q, b\in \Sigma, T \in \{ {\rm S}, {\rm R}\}, \alpha \in \{0,1\}$, we add the following transition to $\delta_{M'}$:
\begin{enumerate}
\item \label{con1} $\delta_{M'}( (q,0,0 ), b, 0) = ((p, \alpha, 0), T, 0) $.
\end{enumerate}

Also, if $\delta(q,b, 1) = (p,T, \alpha)$, for some $q,p \in Q, T \in \{ {\rm S}, {\rm R}\}, \alpha \in \{-1,0,1\}$
we add the following transition in $\delta_{M'}$
for every $s = (q, d, j) \in Q_{M'}$:
\begin{enumerate}

\setcounter{enumi}{1}

\item \label{con2} $\delta_{M'}( s, b, 0) = ((p,d+\alpha, 0),T, 0)$  if  $j = 0, 0 \leq d \leq t$, and $0 \leq d + \alpha \leq t$,

\item \label{con4} $\delta_{M'}( s, b, y) = ((p,t, j + \alpha \mod l),T, \alpha)$ for $y \in \set{0,1}$, if $d = t$ and $\alpha \in \{0,1\}$,

\item \label{con5} $\delta_{M'}( s, b, 1) = ((p,t, j-1 \mod l), T,\alpha)$ if $d = t$ and $\alpha = -1$.

\end{enumerate}

\begin{claim} \label{claim1}
For all $m \in \natzero$,
if $(q_0, w=uv, 0) \vdash_{M}^m (q, v, c)$ where $u,v \in \Sigma^*, q \in Q, c \in \mathbb{N}_0$,
then
\begin{equation}
((q_0,0,0), uv, 0) \vdash_{M'}^m
 ((q, t, c-t \mod l), v, c-t),
\label{claim1eq1}
 \end{equation} when $c>t$,
and
\begin{equation} ((q_0,0,0), uv, 0) \vdash_{M'}^m
((q, c, 0), v, 0),
\label{claim1eq2}
\end{equation} when $c \leq t$.
\end{claim}
\begin{proof}
We perform induction on $m$.

If $m=0$ then $q = q_0, u = \lambda, c = 0, c\leq t$, thus the second condition is true.

Consider $m \geq 0$, and assume the implication holds for $m$. We will show it holds for $m+1$.

Suppose $(q_0, uv, 0) \vdash_{M}^{m+1} (q, v, c)$. Then for some state $p\in Q$, $a \in \Sigma \union \set{\lambda},$
and $c' \in \natzero{}$, we have
$(q_0, uv, 0) \vdash_{M}^m (p, av, c') \vdash_M^1 (q,v,c)$, with the last
transition via $x$.
We know that $c \in \set{c'-1, c', c'+1}$.

\textbf{Case: $c \geq t, c' \geq t$}. Notice that when $c' = t$,
Equations (\ref{claim1eq1}) and (\ref{claim1eq2}) coincide. 
Then, by our hypothesis, we have
$$((q_0,0,0), uv, 0) \vdash_{M'}^m ((p, t, c' - t \mod l), av, c'-t).$$

If $c = c'-1$ (so $c' > t$ and $c' - t>0$), then we know that
$((p, t, c' - t \mod l), av, c'-t) \vdash_{M'} ((q, t, c' - t - 1 \mod l), v, c'-t-1) = ((q, t, c - t \mod l), v, c-t)$
by the transition created by rule (4) from $x$.

If $c = c'$, then we know that
$((p, t, c' - t \mod l), av, c'-t) \vdash_{M'} ((q, t, c' - t \mod l), v, c'-t) = ((q, t, c - t \mod l), v, c-t)$
by transition rule (3).

If $c = c' + 1$, then we know
$((p, t, c' - t \mod l), av, c'-t) \vdash_{M'} ((q, t, c' - t + 1 \mod l), v, c'-t+1) = ((q, t, c - t  \mod l), v, c-t)$
by transition rule (3).

\textbf{Case: $c \leq t, c' \leq t$}.
By our hypothesis we have
$$((q_0,0,0), uv, 0) \vdash_{M'}^m ((p, c', 0), av, 0),$$ $a \in \Sigma \cup \{\lambda\}$.

If $c = c' -1$, $c = c'$ or $c = c'+1$, then the implication holds by transition rule (2), unless
$c' = 0$, in which case it holds by transition rule (1).

Thus we have shown that the implication true for $M'$ in $m+1$ steps, and is therefore true for all $m$.
\qed \end{proof}

\begin{claim} \label{claim2}
For all $m \in \natzero$,
let
$$((q_0,0,0), w=uv, 0) \vdash_{M'}^m ((q, d, j), v, e),$$
where $u,v \in \Sigma^*$.
Then the following are true:
 \begin{enumerate}
\item \label{item1} $d+ j = \delta_D(e+d)$,
\item \label{item4} $(q_0, uv, 0) \vdash_{M}^m (q, v, e+d)$,
\item \label{item2} $e> 0$ or $j>0$ only if $d=t$.
\end{enumerate}
\end{claim}
\begin{proof}
We perform induction on $m$.

If $m=0$ then $d = 0 < t, e= 0, j =0$, and thus (3) is true,
and conditions (1) and (2) are immediate.

Consider $m \geq 0$, and assume the implication holds for $m$.

Suppose $((q_0,0,0), uv, 0) \vdash_{M'}^{m+1} ((q, d, j), v, e)$.
Then
\begin{eqnarray*}
((q_0,0,0), uv, 0) &\vdash_{M'}^{m}& ((q', d', j'), av, e')\\
&\vdash_{M'}&  ((q, d, j), v, e),
\end{eqnarray*} $a \in \Sigma \cup \{\lambda\}$, by some last transition $x$.
Then, by the hypothesis,
$d' + j' = \delta_D(e' + d'), (q_0, uv, 0) \vdash_M^m (q', av, e' + d')$,
and $e'>0$ or $j'>0$ only if $d' =t$. 

Suppose $e'>0$. Then $x$ must be of type (3) or (4) in the
construction, and $d' = t = d$, therefore the third condition
is true. Then the transition $x$ that changes the counter by
$\alpha$ is created from a transition that changes the counter
of $M$ similarly. Thus, the second condition holds. For the first condition, $\delta_D(e+d)$ must be in the ``loop'' of $D(q_i)$ since
$t=d$, and $\delta_D(e+d) = \delta_D(e' + d' + \alpha) =
 d + (j' + \alpha \mod l) = d + j$.

Suppose $e' = 0$. Then $x$ must be of type (1), (2), or (3). If it
is type (3), then $d' = t = d$ (here $\alpha \in \{0,1\}$), then
the conditions hold just like the case above. For
both types (1) and (2), then $e' = e = j = j' = 0$, and so
condition 3 is true. For both, $d'$ changes
to $d$ in the same way as the counter of $M$. Then
the second condition holds. For the first condition, 
$\delta_D(e+d) = \delta_D(d) = d = e+d$.
\qed \end{proof}

Then, $w$ is accepted by final state in $M$, if and only if
$(q_0, w \lhd, 0) \vdash_M^* (q, \lhd, c) \vdash_M^* (q_f, \lhd, c')$, for some $q \in Q, q_f \in F$, and $(q,\lhd, c)$ is the first configuration to reach $\lhd$, if and only if 
$(q_0, w\lhd, 0) \vdash_M^* ( q, \lhd, c)$, for some $q \in Q, (q, \lhd,c)$ is the first configuration to reach $\lhd$, and $a^c \in L(q)$ (from
the definition of $L(q)$), if and only if
$(q_0, w, 0) \vdash_M^* (q, \lambda,c)$ for some $q \in Q$ such that $a^c \in L(q)$. We will show that this is true if and only if $M'$ accepts $w$ by final state without end-marker.

Assume $(q_0, w, 0) \vdash_M^* (q, \lambda,c)$ for some $q \in Q$ such that $a^c \in L(q)$. Then $\delta_D(c)$ is final in $D(q)$, and
$$((q_0,0,0),w,0) \vdash_{M'}^*
(q,f,j),\lambda,e),$$ for some $f,e,j$ by Claim
\ref{claim1}, where $c>t$ implies $f= t, e = c-t$, and $j = c-t \mod l$,
and $c\leq t$ implies $f = c, j = e = 0$. In the second
case, it is immediate that $(q,c,0)$ is final since $a^c \in L(q)$. In the
first case, it follows since $j = c-t \mod l$ and by the structure
of $D(q)$ that 
$\delta_D(c) = t + j$.
Then $(q,f,j)$ is final in $M'$
and $M'$ accepts $w$ by final state without end-marker. 

Conversely, assume
$M'$ accepts $w$ by final state without end-marker. Then
$((q_0,0,0),w,0)\vdash_{M'}^*
((q,f,j),\lambda,e)$, either $f = t$ and $t+j \in F_{D(q)}$,
or $f<t, j = 0$ and $f \in F_{D(q)}$.
Then $(q_0, w, 0) \vdash_M^* (q, \lambda, e+f)$ with
$\delta_D(e+f) = f+j$ by Claim \ref{claim2}. 
If $f = t$, then $a^{e+f} \in L(q)$, and we are done.
If $f<t$, then $e = f = 0$ by Claim \ref{claim2},
$\delta_D(e+f) = \delta_D(f) = f \in F_{D(q)}$, and
$a^{e+f} \in L(q)$.
Thus,
$(q_0,w,0)\vdash_M^* (q, \lambda, e+f)$ for some $1 \leq i \leq n$ such
that $a^{e+f} \in L(q)$.

Hence, $w$ is accepted by final state in $M$ if and only if $w$ is accepted by final state without end-marker in $M'$.
\qed \end{proof}

We will extend these closure results with a lemma about prefix-free $\DCMNE$ languages.
It  is known that a regular language is prefix-free
if and only if there is a non-exiting DFA accepting the language \cite{generalGeneral}. 
\begin{lemma}
Let $L \in \DCMNE$. Then $L$ is prefix-free if and only if
there exists a $\DCM$-machine $M$ accepting $L$ by final state without end-marker which is non-exiting.
\end{lemma}
\begin{proof}
Let $L \in \DCMNE$, with $M$ a machine accepting $L$ by final state without end-marker. Then
$\{w \mid (q_0, w, 0) \vdash_M^* (q_f, \lambda, c), q_f \in F\}=L$.
Assume without loss of generality that no stay transitions can switch to a final state, because a word can only be accepted by final state without
end-marker after a transition moving right. Thus, any stay transition switching
to a final state $q_f$ can switch to a non-final state $q_f'$ that operates 
just like $q_f$.

$(\implies)$
Suppose $L$ is prefix-free. Construct $M'$ from $M$ such that all transitions out of final states are removed, and so $M'$ is not non-exiting. Then $L(M') \subseteq L(M)$ since all transitions of $M'$ are in $M$.
For the reverse containment, consider $w \in L(M)$ such that
$(q_0, w_0, i_{0,1}, \ldots, i_{0,k}) \vdash_M \cdots \vdash_M (q_n, w_n,i_{n,1}, \ldots, i_{n,k}), n \geq 0, q_n \in F, w_n = \lambda, w_0 = w, i_{0,1} = \cdots i_{0,k} = 0$, via transitions $\alpha_1, \ldots, \alpha_n$
respectively.
Assume that there exists $j < n$ such that $q_j \in F$. Thus, $n>0$. If $j>1$,
then $\alpha_{j-1}$ must be a right transition instead of a stay transition,
as no stay transition switches to a final state. But then
the sequence $\alpha_1, \ldots, \alpha_{j-1}$ (or the empty sequence if $j=1$),
is the computation accepting the right quotient of $w$ with $w_j$, which is a proper prefix of
$w$ since $j<n$ and since $\alpha_n$ must be a right transition. But $L$ is prefix-free, a contradiction. Thus, all of $q_0, \ldots, q_{n-1}$ are non-final,
and $\alpha_1, \ldots, \alpha_n$ are in $M'$.
Thus $L(M)\subseteq L(M')$ as well.

$(\impliedby)$
Suppose $M$ is non-exiting. Consider $w\in L$.
Then after reading $w$ deterministically, there are no transitions to follow, so $wx$ is not accepted for any $x \neq \lambda$.
Thus $L$ is prefix-free.

\qed \end{proof}


From this, we obtain a special case where $\DCM$ is closed under concatenation, if the first language can be both accepted by final state without end-marker, and is prefix-free. The construction considers a non-exiting machine accepting $L_1$ by final state without end-marker, where transitions into its final state are replaced by transitions into the initial state of the machine accepting $L_2$. 
\begin{theorem}
\label{concatenationprefixfreeDCM}
Let $L_1 \in \DCMNE(k,l), L_2 \in \DCM(k',l')$, with $L_1$ prefix-free.
Then $L_1 L_2 \in \DCM(k+k', \max(l,l') )$.
\end{theorem}
\begin{proof}
Our construction is simple. Consider non-exiting $M_1$ accepting $L_1$ by final state without end-marker,
and $M_2$ accepting $L_2$. Assume without loss of generality that only
transitions that move right in $M_1$ switch to a final state.
We form $M'$ where $L(M') = L_1 L_2$. Indeed,
$M'$ has the states and transitions from $M_1, M_2$ combined, with the start state of $M_1$ as its start state.
Any transition into an accepting state of $M_1$ is replaced by an equivalent transition
into the starting state of $M_2$.
The accepting states are the accepting states of $M_2$.
The machine has separate counters for the counters of $M_1$ and $M_2$, each of which performs
the same reversals they would in their original machine.

Let $w \in L_1$ and $x \in L_2$. Since $M_1$ accepts without end-marker,
and since no proper prefix of $w$ leads to an accepting state of $M_1$,
we know reading $w$ in $M_1$ leads to an accepting state in $M_1$, even without reading $\lhd$.
So, in $M'$, we know that reading $w$ will lead to the start state of $M_2$.
Reading $x$ from the start of $M_2$ leads to an accepting state, since $x \in L_2$.
Thus reading $x$ from the start of $M_2$ in $M'$ leads to acceptance.

Let $y \in L(M')$. Then
$M'$ starts in the start of $M_1$, so the only path to an accepting state is through the start
of $M_2$. Thus there is some division of $y$ into $w,x$ where reading $w$ in $M_1$ leads to acceptance
(because it leads to the start state of $M_2$ in $M'$), and reading $x$ in $M_2$ leads to acceptance,
because we got to an accepting state in $M'$. Thus $y \in L_1 L_2$.
\qed \end{proof}
Notice that it is also possible to make 
$L_1 L_2 \in \DCM( \max\{k,k'\}, l+l' + 1)$ by resetting and
reusing the same counters for $M_1$ and $M_2$.

If we remove the condition that $L_1$ is prefix-free however, the theorem is no longer true, as we will see in the next section that even the regular language $\Sigma^*$ (which is in $\DCMNE(0,0)$) concatenated with a $\DCM$ language produces a language outside $\DCM$.

\begin{corollary}
\label{prefixfreeregular}
Let $L \in \DCM(k,l), R \in \REG$, where $R$ is prefix-free.
Then $RL \in \DCM(k,l)$.
\end{corollary}

In contrast to left concatenation of a regular language with a $\DCM$ language (Corollary \ref{prefixfreeregular}), where it is required that $R$ be prefix-free (the regular language is always in $\DCMNE$), for right concatenation, it is only required that it be a $\DCMNE$ language. We will see in the next section that this is not true if the restriction that $L$ accepts by final state without end-marker is removed.

The following proof takes a $\DCM$ machine $M_1$ accepting by final state without end-marker, and $M_2$ a DFA accepting $R$, and builds a $\DCM$ machine $M'$ accepting $LR$ by final state without end-marker. 
\begin{theorem}
\label{DCMNEwithREG}
Let $L \in \DCMNE(k,l)$, $R \in \REG$.
Then $LR \in \DCMNE(k,l)$.
Hence, $\pref^{-1}(L) = L \Sigma^* \in \DCMNE(k,l)$.
\end{theorem}
\begin{proof}
Let $M_1=(k,Q_1,\Sigma,\lhd,\delta_1,q_1,F_1)$ be a $\DCM$ machine accepting $L$ by final state without end-marker where, without loss of generality, final
states are only reached after transitions that move right. Let $M_2=(Q_2,\Sigma,\delta_2,q_2,F_2)$ be a DFA accepting $R$. A $\DCM$ machine $M' =(k,Q',\Sigma,\lhd,\delta',q',F')$ will be built that will accept $LR$ by
final state without end-marker. Assume without loss of generality that $M_1$ reads all the way to the end of every input.
This can be assumed, similar to the proof of closure of $\DCM$ under complement \cite{Ibarra1978} by removing the ability of $M_1$ to
enter an infinite loop on the input which can be detected using the finite control, and instead switching to a ``dead state'' while
reading all of the input.

Intuitively, $M'$ will simulate $M_1$ while also storing a subset of $Q_2$ in the second component of the state. Every time it reaches a final state of $M_1$, it places the initial state of $M_2$ in the second component. And, then it continues to simulate $M_1$, while in parallel simulating the DFA $M_2$ on every state in the second component in parallel. 

Formally, $Q' = Q_1 \times 2^{Q_2}$, $q' = (q_1,\emptyset)$ if $q_1 \notin F_1$ and $q' = (q_1,\{q_2\})$ otherwise, $F' = \{(q,X) \mid q\in Q_1, X \cap F_2\neq \emptyset\}$ and $\delta'$ is defined as follows:
for every transition, $\delta_1(q,a,x) = (p,T,i), p,q\in Q_1, a \in \Sigma, x \in \{0,1\}^k, T \in \{{\rm S},{\rm R}\}, i \in \{-1,0,1\}^k$, introduce
$\delta'((q,Y),a,x) = ((p,Z),T,i)$, for all $Y \in 2^{Q_2}$, where
\begin{itemize}
\item $Z = Y  \mbox{~if~} T = {\rm S} \mbox{~(and hence~} p \notin F_1$)
\item $Z = \delta_2(Y,a) \mbox{~if~} T = {\rm R} \mbox{~and~} p \notin F_1$
\item $Z = \delta_2(Y,a) \cup \{q_2\} \mbox{~if~} T = {\rm R} \mbox{~and~} p \in F_1$
\end{itemize}

\begin{claim}
$L(M_1)L(M_2) \subseteq L(M')$.
\end{claim}
\begin{proof}
Let $uv \in \Sigma^*$, where $u \in L(M_1), v \in L(M_2)$. Then there is a computation
$(p_1,u_1, i_1(1), \ldots, i_k(1)) \vdash_M \cdots \vdash_M (p_n, u_n, i_1(n), \ldots , i_k(n))$
where $p_1 = q_1, u_1 = uv, i_1(1)= \cdots = i_k(1) = 0, p_n \in F_1, u_n = v$. Furthermore, since $M_1$ reads every input, $(p_n,u_n, i_1(n), \ldots, i_k(n)) \vdash_{M_1}^* (p', \lambda, i_1, \ldots, i_k)$, for some $p', i_1, \ldots, i_k$.
Then, by the construction, $((p_1,Y_1), uv, 0, \ldots , 0) \vdash_{M'}^* ((p_n,Y_n),v ,i_1(n), \ldots, i_k(n))$, where $Y_1 = \emptyset$ and $q_2 \in Y_n$ since $p_n \in F_1$. Furthermore, it must be the case that
$((p_n,Y_n), v, i_1(n), \ldots, i_k(n)) \vdash_{M'}^* ((p',Y'),\lambda, i_1, \ldots, i_k)$, and that $\hat{\delta}(q_2,v) \in Y'$ since every transition applied to $M_1$ while reading $v$ that consumes an input letter, also changes state via that letter according to the DFA $M_2$. Thus, there is a final state from $F_2$ in $Y'$ causing $M'$ to also accept.
\qed \end{proof}

\begin{claim}
$L(M')\subseteq L(M_1)L(M_2)$
\end{claim}
\begin{proof}
Let $w \in L(M')$. Then,
$$((p_1,Y_1),u_1,i_1(1), \ldots, i_k(1)) \vdash_{M'} \cdots \vdash_{M'} ((p_n,Y_n), u_n, i_1(n), \ldots , i_k(n)),$$ where $u_1 = w, p_1 = q_1, Y_1 = \emptyset, u_n = \lambda, i_1(1)= \cdots = i_k(1)=0, Y_n \cap F_2 \neq \emptyset$. Let $q_f$ be some state in $F_2 \cap Y_n$.
Then, by the construction, there exists some $j$, $1 \leq j \leq n$ such that $p_j \in F_1, q_2 \in Y_j$,
and for every transition from the $j$th configuration to the last one, while reading $u_j$,
the sets $Y_j, \ldots , Y_n$ consecutively stay the same on a stay transition, and on a right transition
 that consumes the next input letter of $u_j$, puts the state $\hat{\delta_2}(q_2,u_j')$, for each consecutive prefix $u_j'$ of $u_j$ in the sets $Y_j, \ldots, Y_n$. Hence, $u_j \in R$, and since $p_j \in F_1$, it must be that $w u_j^{-1}\in L(M_1)$.
\qed \end{proof}
Hence, $LR \in \DCMNE(k,l)$.
\qed \end{proof}

As a corollary, we get that $\DCM(1,l)$ is closed under right concatenation with regular languages. This corollary could also be inferred from the proof  that deterministic context-free languages are closed under concatenation with regular languages \cite{harrison1978}.
\begin{corollary}
\label{oneCounterConcat}
Let $L \in \DCM(1,l)$ and $R \in \REG$. Then $LR \in \DCM(1,l)$.
\end{corollary}

\begin{corollary}
\label{dcminverseprefix}
If $L \in \DCM(1,l)$, then $\pref^{-1}(L) \in \DCM(1,l)$.
\end{corollary}

\section{Relating (Un)Decidable Properties to Non-closure Properties}
\label{sec:nonclosure}

In this section, we use a technique that proves non-closure properties
using (un)decidable properties. 
 A similar technique was used in \cite{Chiniforooshan2012} for showing that there is a language accepted by
a 1-reversal DPDA that cannot be accepted by any $\NCM$.
In particular, we use this technique
to prove that some languages are not in both $\DCM$ and $2\DCM(1)$ (i.e., accepted by two-way
DFAs with one reversal-bounded counter). Since $2\DCM(1)$s
have two-way input and a reversal-bounded counter, it does not seem
easy to derive ``pumping'' lemmas for these machines.
$2\DCM(1)$s are quite powerful, e.g.,
although the Parikh map of the language accepted by
any finite-crossing $2\NCM$ (hence by  any $\NCM$) is semilinear
\cite{Ibarra1978},
$2\DCM(1)$s can accept non-semilinear languages.  For example,
$L_1= \{a^i b^k ~|~ i, k \ge 2, i$ divides $k \}$ can
be accepted by a $2\DCM(1)$ whose counter makes only one reversal.
This technique is used to establish that the inverse infix,
inverse suffix, and inverse outfix closure
of a language in $\DCM(1,1)$ can be outside of both $\DCM$ and $2\DCM(1)$. It is also used to show that the
inverse prefix closure of a $\DCM(2,1)$ language can be
outside of both $\DCM$ and $2\DCM(1)$.

We will need the following result
(the proof for DCM is in \cite{Ibarra1978};
the proof for 2DCM(1) is in \cite{IbarraJiang}):

\begin{proposition} \label{thm1}
$~~$
\begin{enumerate}
\item
The class of languages $\DCM$ is
closed under Boolean operations.  Moreover, the emptiness
problem is decidable.
\item
The class of languages $2\DCM(1)$ is
closed under Boolean operations.  Moreover, the emptiness
problem is decidable.
\end{enumerate}
\end{proposition}

\noindent
We note that the emptiness problem for $2\DCM(2)$,
even when restricted to machines accepting only
letter-bounded languages (i.e., subsets of
$a_1^* \cdots a_k^*$ for some $k \ge 1$ and
distinct symbols $a_1, \ldots, a_k$)
is undecidable \cite{Ibarra1978}.

We will show that there is a language $L \in \DCM(1,1)$
such that $\infx^{-1}(L)$ is not in $\DCM \cup 2\DCM(1)$.

The proof uses the fact that that there is a recursively enumerable
language $L_{\rm re} \subseteq \natzero$ that is not recursive
(i.e., not decidable) which is accepted by a deterministic
2-counter machine \cite{Minsky}.  Thus, the machine
when started with $n \in \natzero$ in the first counter
and zero in the second counter,
eventually halts (i.e., accepts $n \in L_{\rm re}$).

A close look at the constructions  in \cite{Minsky}
of the 2-counter machine, where initially one counter has some
value $d_1$ and the other counter is zero, reveals
that the counters behave in a regular pattern. The 2-counter machine
operates in phases in the following way.
The machine's operation can be divided into phases, where each
phase starts with one of the counters equal to
some positive integer $d_i$ and the other counter equal to 0.
During the phase, the positive counter decreases, while the other
counter increases. The phase ends with the first counter having
value 0 and the other counter having value $d_{i+1}$.
Then in the next phase the modes of the counters are interchanged.
Thus, a sequence of configurations corresponding to the phases
will be of the form:

\vskip .25cm

$(q_1, d_1, 0), (q_2, 0, d_2), (q_3, d_3, 0), (q_4, 0, d_4), (q_5, d_5, 0), (q_6, 0, d_6), \dots$

\vskip .25cm

\noindent
where the $q_i$'s  are states, with $q_1 = q_s$ (the
initial state), and $d_1, d_2, d_3, \ldots$ are positive
integers. Note that in going from state $q_i$ in phase $i$ to
state $q_{i+1}$ in phase $i+1$, the 2-counter machine goes
through intermediate states.
Note that the second component of
the configuration refers to the value of $c_1$ (first counter), while the third
component refers to the value of $c_2$ (second counter).

For each $i$, there are 5 cases for the value of $d_{i+1}$
in terms of $d_i$:
$d_{i+1} = d_i, 2d_i, 3d_i, d_i/2, d_i/3$.
(The division operation is done only if the number is divisible
by 2 or 3, respectively.)
The case is determined by $q_i$.
Thus, we can define a mapping $h$ such if $q_i$
is the state at the start of phase $i$,
$d_{i+1} = h(q_i)d_i$ (where $h(q_i)$ is either
1, 2, 3, 1/2, 1/3).

Let $T$ be a 2-counter machine accepting
a recursively enumerable set $L_{\rm re}$ that is not recursive.  We assume
that $q_1=q_s$ is the initial state, which is never re-entered,
and  if $T$ halts, it does so in a unique state $q_h$.
Let $T$'s state set be $Q$, and $1$ be a new symbol.

In what follows, $\alpha$ is any sequence of
the form $\# I_1 \# I_2 \# \cdots \# I_{2m} \#$ (thus we assume that
the length is even), where
$I_i = q1^k$ for some $q \in Q$ and $k \ge 1$, represents
a possible configuration of $T$ at
the beginning of phase $i$, where $q$ is the state and
$k$ is the value of counter $c_1$ (resp., $c_2$) if $i$
is odd (resp., even).

\vskip .25cm

\noindent
Define $L_0$ to be the set of all strings $\alpha$ such that
\begin{enumerate}
\item $\alpha = \#I_1 \#I_2\# \cdots \#I_{2m}\#$;
\item $m \ge 1$;
\item for $1 \le j \le 2m-1$, $I_j \Rightarrow I_{j+1}$, i.e.,
if $T$ begins in configuration $I_j$, then after one phase,
$T$ is in configuration $I_{j+1}$ (i.e., $I_{j+1}$
is a valid successor of $I_j$);
\end{enumerate}

\begin{lemma} \label{lem1}
$L_0$ is not in $\DCM \cup 2\DCM(1)$.
\end{lemma}
\begin{proof}
Suppose $L_0$ is accepted by a $\DCM$ (resp., $2\DCM(1)$).
The following is an algorithm to decide, given
any $n$, whether $n$ is in $L_{\rm re}$.

\begin{enumerate}
\item
Let $R = \#q_s1^n ((\#Q1^+\#Q1^+))^*\#q_h1^+\#$.
Then $R$ is regular.
\item
Then $L' = L_0 \cap R$ is also in $\DCM$ (resp., $2\DCM(1)$)
by Proposition \ref{thm1}.
\item
Check if $L'$ is empty. This is possible, since
emptiness of $\DCM$ (respectively, $2\DCM(1)$) is decidable by Proposition \ref{thm1}.
\end{enumerate}

\vskip .25cm

\noindent
The claim follows, since $L'$ is empty if and only if
$n$ is not in $L_{\rm re}$.
\qed \end{proof}

\subsection{Non-closure Under Inverse Infix}

\begin{theorem}
\label{prop:inverseNotInDCM}
There is a language $L \in \DCM(1,1)$
such that $\infx^{-1}(L) = \Sigma^* L \Sigma^*$ is not in
$\DCM \cup 2\DCM(1)$.
\end{theorem}
\begin{proof}
Let $T$ be a 2-counter machine.
Let $L = \{\#q1^m \#p1^n\# \mid \mbox{in~} T, q1^m\not\Rightarrow p1^n \}$.
That is, $L$ contains all pairs of configurations of $T$ where,
when starting
in state $q$ with $m$ on one counter and zero on the other,
at the next phase, $T$ does not reach state $p$
with the first counter empty, and $n$ in the second counter.
Thus, $L = \{\# I \# I' \# ~|~ I$ and $I'$ are configurations
of $T$, and $I'$ is not a valid successor of $I \}$.
Since $T$ is a deterministic counter machine that, within one phase,
only decreases one counter while increasing another,
$L \in \DCM(1,1)$ since the input tape of the $\DCM(1,1)$ machine
can be used to simulate the decreasing counter (by reading the first configuration) while using the counter to simulate the increasing counter, then verifying that the configuration reached does not match the second input configuration.

We claim that $L_1 = \infx^{-1}(L)$ is not
in $\DCM \cup 2\DCM(1)$.
Otherwise, by Proposition \ref{thm1}, $\overline{L_1}$ (the complement of $L_1$)
is also in $\DCM \cup 2\DCM(1)$, and
$\overline{L_1} \cap (\#Q1^+\#Q1^+)^+\# = L_0$ would be in
$\DCM \cup 2\DCM(1)$.
This contradicts Lemma \ref{lem1}.
\qed \end{proof}

%


\subsection{Non-closure Under Inverse Prefix}

\begin{theorem} \label{nonclosure2}
There exists a language $L$ such that $L  \in \DCM(2,1)$ and
$L \in 2\DCM(1)$ (accepted by a two-way machine that makes one turn on the input tape and the counter is
1-reversal-bounded) such that
$\pref^{-1}(L) = L \Sigma^* \not \in \DCM \cup 2\DCM(1)$.
\end{theorem}
\begin{proof}
Consider $L = \{ \# w \# \sst$
$w \in \set{a,b,\#}^*, \abs{w}_a \neq \abs{w}_b \}$.
Then $L \in \DCM(2,1)$, as a machine can be built that
records the number of $a$'s and $b$'s in two counters, and then
once it hits the end-marker, subtracts both in parallel to verify that they
are different (it can also be accepted
by a $2\DCM(1)$ machine that records the number of $a$'s,
then makes a turn on the input and verifies that the number
of $b$'s is different).

Suppose to the contrary that $\pref^{-1}(L) \in \DCM \cup 2\DCM(1)$.
Then, $L' \in \DCM \cup 2\DCM(1)$,
where $L' = \pref^{-1}(L) \cap (\# \set{a,b,\#}^* \#) =$
$\{\# w_1 \cdots \# w_n \# \sst$
$ \exists i . \abs{w_1 \cdots w_i}_a \neq \abs{w_1 \cdots w_i}_b \}$.

Let $L'' = \overline{L'} \cap (\# a^* b^*)^+ \#$. It follows that $L''$
is in $\DCM$ and $2\DCM(1)$ since both are closed under complement
and intersection with regular languages \cite{Ibarra1978}.
Then $L'' \in \DCM \cup 2\DCM(1)$.
Further, $L'' = \set{\# a^{k_1} b^{k_1}\# \cdots \# a^{k_m} b^{k_m} \# \sst m > 0}$.

We will show that $L''$ is not in $\DCM \cup 2\DCM(1)$, which
will lead to a contradiction.
Define two languages:

\begin{itemize}
\item
$L_1 = \{\#1^{k_1}\#1^{k_1}\# \cdots \#1^{k_m}\#1^{k_m}\#  ~|~ m \ge 1, k_i \ge 1 \}$,
\item
$L_2 = \{\#1^{k_0}\#1^{k_1}\#1^{k_1}\#\cdots \#1^{k_{m-1}}\#1^{k_{m-1}}
\#1^{k_m}\# ~|~ m \ge 1,  k_i \ge 1 \}$.
\end{itemize}

\vskip .25cm

\noindent
Note that $L_1$ and $L_2$ are similar. In $L_1$, the odd-even
pairs of blocks 1's are the same, but in $L_2$, the even-odd pairs of 
blocks of 1's are
the same. If $M''$ accepts $L''$ in $\DCM \cup 2\DCM(1)$,
then it is possible to construct (from $M''$) $M_1$ and $M_2$ in $\DCM \cup 2\DCM(1)$ to
accept $L_1$ and $L_2$, respectively.

We now refer to the language $L_0$ that was shown not to be in
$\DCM \cup 2\DCM(1)$ in Lemma \ref{lem1}.
We will construct a DCM (resp., 2DCM(1)) to accept $L_0$,
which would be a contradiction.
Define the languages:
\begin{itemize}
\item
$L_{odd} = \{\#I_1\#I_2\# \cdots \#I_{2m} ~|~ m \ge 1,
I_1, \cdots, I_{2m}$ are configurations of the
2-counter machine $T$, for odd $i$, $I_{i+1}$ is
a valid successor of $I_i \}$.
\item
$L_{even} = \{\#I_1\#I_2\# \cdots \#I_{2m} ~|~ m \ge 1,
I_1, \cdots, I_{2m}$ are configurations of the
2-counter machine $T$, for even $i$, $I_{i+1}$ is
a valid successor of $I_i \}$.
\end{itemize}

\noindent
Then $L_0 = L_{odd} \cap L_{even}$.
Since DCM (resp., 2DCM(1)) is closed
under intersection, we need only to construct two DCMs (resp., 2DCM(1)s)
$M_{odd}$ and
$M_{even}$ accepting $L_{odd}$ and $L_{even}$, respectively.
We will only describe the construction of $M_{odd}$, the
construction of $M_{even}$ being similar.

\vskip .25cm

\noindent
{\bf Case:}  Suppose $L'' \in \DCM$:

\noindent
First consider the case of DCM.  We will construct two machines:
a DCM $A$ and a DFA $B$ such that $L(M_{odd}) = L(A) \cap L(B)$.

Let $L_A = \{\#I_1\#I_2\# \cdots \#I_{2m} ~|~ m \ge 1,
I_1, \cdots, I_{2m}$ are configurations of the
2-counter machine $T$, for odd $i$, if $I_i  = q_i1^{d_i}$,
then $d_{i+1} = h(q_i)d_i \}$.
We can construct a DCM $A$ to accept $L_{A}$
by simulating the DCM $M_1$.
For example, suppose $h(q_i) =3$. Then $A$
simulates $M_1$ but whenever $M_1$ moves its input
head one cell, $A$ moves its input head 3 cells.
If $h(q_i) = 1/2$, then when $M_1$ moves its head
2 cells, $A$ moves its input head 1 cell.
(Note that $A$ does not use the 2-counter machine $T$.)

Now Let $L_B = \{\#I_1\#I_2\# \cdots \#I_{2m} ~|~ m \ge 1,
I_1, \cdots, I_{2m}$ are configurations of the
2-counter machine, for odd $i$, if $I_i = q_i1^{d_i}$,
then $T$ in configuration $I_i$ ends phase $i$
in state $q_{i+1} \}$.  Then, a DFA $B$ can accept
$L_B$ by simulating $T$ for each odd $i$ starting
in state $q_i$ on $1^{d_i}$ {\em without} using a counter,
and checking that the phase ends in state $q_{i+1}$.
(Note that the DCM $A$ already checks the ``correctness''
of $d_{i+1}$.)

We can then construct from $A$ and $B$ a DCM $M_{odd}$
such that $L(M_{odd}) = L(A) \cap L(B)$.
In a similar way, we can construct $M_{even}$.

\vskip .25cm

\noindent
{\bf Case:}  Suppose $L'' \in 2DCM(1)$:

\noindent
The case 2DCM(1) can be shown similarly.
For this case, the machines $M_{odd}$ and $M_{even}$
are 2DCM(1)s, and machine $A$ is a 2DCM(1),
but machine $B$ is still a DFA.
\qed \end{proof}

The language $L$ in the proof above can be accepted by a $\DCM(2,1)$ machine that uses the end-marker. However, we
see next that this language $L$ cannot be accepted by any $\DCMNE$ machine.
\begin{corollary}
There are languages in $\DCM(2,1)$ that are not in $\DCMNE$.
\end{corollary}
\begin{proof}
Consider the language $L$ from the proof of Theorem \ref{nonclosure2}. This theorem shows $L \in \DCM(2,1)$, but that $L \Sigma^* \notin \DCM$, which therefore
implies $L\Sigma^* \notin \DCMNE$. Suppose, by contradiction that $L \in \DCMNE$.
But, $\DCMNE$ is closed under concatenation with $\Sigma^*$ by Theorem \ref{DCMNEwithREG}, and therefore
$L\Sigma^* \in \DCMNE$, a contradiction.
\qed \end{proof}

Hence, the right end-marker is necessary for deterministic counter machines when there are at least two $1$-reversal-bounded counters. In fact, without it, no amount of reversal-bounded counters with a deterministic machine could accept even some languages that can be accepted with two $1$-reversal-bounded counters could with the end-marker.

Furthermore, if $L$ is a $\DCM$ language, then $L\$$ ($\$$ a new symbol)
is in $\DCMNE$. Therefore, if $\DCMNE$ were closed under right quotient
with a single symbol, then $\DCM$ would be equal to $\DCMNE$ which
is not true. Thus, the following result is obtained.
\begin{corollary}
$\DCMNE$ is not closed under right quotient with a single symbol.
\end{corollary}
This is in contrast to $\DCM$ which is closed under right quotient
with context-free languages \cite{deletionDCM}, but requires
the end-marker for this proof, and therefore the end-marker
cannot be removed.

\subsection{Non-closure for Inverse Suffix, Outfix and Embedding}

\begin{theorem}
\label{inversesuffix}
There exists a language $L \in \DCM(1,1)$ such that $\suff^{-1}(L) \not \in \DCM$
and $\suff^{-1}(L) \not \in \TwoDCM(1)$.
\end{theorem}
\begin{proof}
Let $L$ be 	as in Theorem \ref{prop:inverseNotInDCM}.
We know $\DCM(1,1)$ is closed under $\pref^{-1}$ by Corollary \ref{dcminverseprefix}, so
$\pref^{-1}(L) \in \DCM(1,1)$.
Suppose $\suff^{-1}(\pref^{-1}(L)) \in \DCM$.
This implies that $\inf^{-1}(L) \in \DCM$, but we showed this language was not in $\DCM$.
Thus we have a contradiction.
A similar contradiction can be reached if we assume $\suff^{-1}(\pref^{-1}(L)) \in \TwoDCM(1)$.
\qed \end{proof}

\begin{corollary}
\label{leftreg}
There exists $L \in \DCM(1,1)$ and regular language $R$ such that $RL \notin \DCM$ and $RL \notin \TwoDCM(1)$.
\end{corollary}

This implies that without the prefix-free condition on $L_1$ in Theorem \ref{concatenationprefixfreeDCM}, concatenation closure does not follow.
\begin{corollary}
\label{leftdcmne}
There exists $L_1 \in \DCMNE(0,0)$ (regular), and $L_2 \in \DCM(1,1)$, where $L_1L_2 \notin \DCM$ and $L_1L_2 \notin 2\DCM(1)$.
\end{corollary}

The result also holds for inverse outfix.
\begin{theorem}
\label{inverseoutfix}
There exists a language $L \in \DCM(1,1), L \subseteq \Sigma^*$ 
such that $\outf^{-1}(L) \not \in \DCM$ and $\outf^{-1}(L) \not \in 2\DCM(1)$, where $\outf^{-1}(L)\subseteq (\Sigma \cup \{\$\})^*$.
\end{theorem}
\begin{proof}
Consider $L \subseteq \Sigma^*$ where $L \in \DCM(1,1)$, and
$\suff^{-1}(L) \not \in \DCM$ and $\suff^{-1}(L) \not \in 2\DCM(1)$.
The existence of such a language is guaranteed by Theorem \ref{inversesuffix}.
Let $\Gamma = \Sigma \cup \set{\%}$.

Suppose $\outf^{-1}(L) \in \DCM$ over $\Gamma^*$. Then $L' \in \DCM$,
where $L' = \outf^{-1}(L) \cap \% \Sigma^* $.
We can see $L' = \set{\% yx \sst x \in L, y \in \Sigma^*}$,
since the language we intersected with ensures that the section is always
added to the beginning of a word in $L$.

However, we also have $\%^{-1}L' \in \DCM$ because $\DCM$ is closed under left quotient with a fixed word (this can be seen
by simulating a machine on that fixed word before reading any input
letter).
We can see $\%^{-1}L' = \set{yx  \sst x \in L, y \in \Sigma^*}$.
This is just $\suff^{-1}(L)$, so $\suff^{-1}(L) \in \DCM$, a contradiction.

The result is the same for $2\DCM(1)$, relying on the closure of the family under left quotient with a fixed word, which can
be shown be shown by simulating the symbol to be removed on the left
input end-marker.
\qed \end{proof}

\begin{corollary}
\label{inverseembedding}
Let $m \in \natnum$. There exists a language $L \in \DCM(1,1),
L \subseteq \Sigma^*$ such that $\emb^{-1}(m,L) \not \in \DCM$ and $\emb^{-1}(m,L) \not \in 2\DCM(1)$, where $\emb^{-1}(m,L)\subseteq
(\Sigma \cup \{\#, \%\})^*$.
\end{corollary}
\begin{proof}
Consider $L$ as in Theorem \ref{inverseoutfix} above,
and let $\Gamma = \Sigma \cup \{\#, \%\}$.
Let $\emb^{-1}(m,L)$ over $\Gamma^*$.
Then
$$\emb^{-1}(m,\#^{m}L) \cap (\#\%)^{m}\Sigma^* = \{(\#\%)^{m} yx \mid x\in L, y \in \Sigma^*\},$$ 
since this enforces that all $m$-embedded words are of the form $\%\#$ except the $m$'th, which may
also insert an arbitrary $y \in \Sigma^*$ before $x \in L$. The rest proceeds just like 
Theorem \ref{inverseoutfix}.
\qed
\end{proof}

\section{Inverse Transducers}

This section studies transducers with reversal-bounded counters and other stores attached.
Using the inverse of such transducers allows for creating elaborate methods
of insertion (such as in Example \ref{example1} below).
It is shown that $\DCM$ is closed under inverse deterministic reversal-bounded multicounter transductions, and $\NCM$ is closed under inverse nondeterministic
reversal-bounded multicounter transductions, and they are both the
smallest family of languages where this holds. Hence, this demonstrates
a method of defining insertion operations under which $\DCM$ is closed
(in contrast to the insertion methods of Section \ref{sec:nonclosure}).

\begin{definition}
A {\em $k$-counter transducer} 
$A =(k,Q,\Sigma,\Gamma, \lhd, \delta, q_0,F)$ where $Q, \Sigma, \Gamma, \lhd, q_0, F$ are respectively the sets
of states, input alphabet, output alphabet, right end-marker 
(not in $\Sigma \cup \Gamma$), initial state $q_0 \in Q$,
and set of final states $F\subseteq Q$. 
The transition function is a partial function from 
$Q \times (\Sigma \cup \{\lhd\}) \times \{0,1\}^k$ into
the family of subsets of
$Q \times \{{\rm R},{\rm S}\} \times \{-1,0,+1\}^k  \times \Gamma^*$.  $M$ is deterministic if every element mapped by $\delta$ is to a subset with one element in it, and if
$\delta(F \times \{\lhd\} \times \{0,1\}^k) = \emptyset$ to prevent
multiple outputs from the same input on deterministic transducers.
A configuration of $A$ is of the form $(q,w\lhd, c_1, \ldots, c_k, z)$,
where $q \in Q$ is the current state, $w\in \Sigma^*$ is the remaining
input, $c_1, \ldots, c_k \in \mathbb{N}_0$ are the counter contents, and 
$z \in \Gamma^*$ 
is the accumulated output. Then,
$(q,aw, c_1,\ldots, c_k, z) \vdash_A (p,w', c_1 +d_1, \ldots, c_k +d_k, z'), a \in \Sigma \cup \{\lhd\}, 
aw,w' \in \Sigma^* \lhd$, where 
$(p,d,d_1, \ldots, d_k,x) \in \delta(q,a,\pi(c_1),\ldots,\pi(c_k)), 
z' = zx, (d = {\rm S} \Rightarrow aw = w')$, and 
$(d = {\rm R} \Rightarrow w = w')$. 
Then $\vdash_A^*$ is the reflexive-transitive closure of $\vdash_A$.
In the definition above, if there are no counters, then $k$ and the counter
contents are left off of the definitions.

Let $L\subseteq \Sigma^*$, and let 
$A= (k,Q,\Sigma,\Gamma,\lhd,\delta,q_0,F)$ be 
a $k$-counter transducer. Then 
$$A(L) = \{x \mid (q_0,w\lhd, 0, \ldots, 0, \lambda) 
\vdash_A^* (q_f, \lhd, c_1, \ldots, c_k, x),w \in L, q_f \in F\}.$$
Let $L \subseteq \Gamma^*$. Then
$$A^{-1}(L) = \{w \mid (q_0, w\lhd,0, \ldots, 0,\lambda) \vdash_A^* 
(q_f,\lhd,c_1, \ldots, c_k,x), x \in L, q_f \in F\}.$$
Also, $A$ is $l$-reversal-bounded if all counters are $l$-reversal-bounded
on input $\Sigma^*$.
\end{definition}

From this definition, the following closure property can be obtained.
\begin{lemma}
\label{invTransductions}
$\DCM$ is closed under inverse deterministic reversal-bounded
counter transductions, and $\NCM$ is closed under inverse reversal-bounded
counter transductions. 
\end{lemma}
\begin{proof}
Let $M = (k,Q,\Gamma, \lhd, \delta,q_0,F)$ be a $k$-counter
$l$-reversal-bounded $\DCM$.
Let $A = (k_A,Q_A,\Sigma,\Gamma,\lhd, \delta_A, q_A, F_A)$ be a deterministic 
$l_A$-reversal-bounded $k_A$-counter transducer.

Then we construct a $\max\{l,l_A\}$-reversal-bounded $\DCM$ machine
$M' = (k+k_A, Q',\Sigma,\lhd, \delta', q_0',F')$ accepting 
$A^{-1}(L(M))$
as follows: 
$M'$ takes as input a word $a_1 \cdots a_n \in \Sigma^*, a_i \in \Sigma$ followed by the end marker $\lhd$. 
In the states of $Q'$, $M'$ keeps a buffer of at most length 
$\alpha = \max\{|x| \mid (p,d,d_1, \ldots, d_k,x) \in \delta_A(q,a,i_1, \ldots, i_k)\} +1$. Then 
on each input letter, $a_i$, $M'$ simulates one transition of $A$ on $a_i$, and
stores the (deterministically calculated) output in the buffer, while using
the first $k_A$ counters. 
If the buffer becomes non-empty, $M'$ simulates
$M$ on the buffer and the remaining $k$ counters. 
Once the buffer becomes empty again, $M'$
continues the simulation of $A$ (on $a_i$ if the transition of $A$ applied
last was a stay transition, and on $a_{i+1}$ if it was a right transition).
If $M'$ reaches the end-marker of $A$, and $A$ is in a final state, then $M'$
puts the end-marker $\lhd$ at the end of the output buffer. If this occurs, then $M'$ continues simulating $M$ on the buffer, accepting if it reaches
a state of $F$ with only $\lhd$ in the buffer.

The proof is similar for $\NCM$.
\qed \end{proof}

This same proof technique can be generalized to other models where stores can be combined without increasing the
capacity. But even when, for example, combining two arbitrary (non-reversal-bounded counters) counters, such
machines already have the full power of Turing machines.

From this, we can immediately get a relatively simple characterization
of $\DCM$ and $\NCM$ languages.
\begin{theorem}
\label{DCMCharacterization}
$L$ is in $\DCM$ ($\NCM$ respectively) if and only if there is a deterministic
(nondeterministic) reversal-bounded counter transducer $A$ such that
$L = A^{-1}(\{\lambda\})$. Hence,
$\DCM$ ($\NCM$ respectively) is the smallest family of languages containing $\{\lambda\}$ that
is closed under inverse deterministic (nondeterministic) reversal-bounded counter transductions.
\end{theorem}
\begin{proof}
Let $M$ be a $\DCM$ machine which, without loss of generality, does not have any transitions defined
on a final state and the end-marker (these can be removed without changing the language accepted). Let $A$ be the 
reversal-bounded multicounter transducer that is obtained from $M$ (same states, 
transitions, and final states), but outputs
$\lambda$ on every transition. Then $A$ is deterministic and 
$A^{-1}(\{\lambda\}) = \{w \mid w \in L(M)\}$. Similarly for $\NCM$.
\end{proof}

A brief example will be given next showing how such a transducer can define an insertion into a $\DCM$ language.
\begin{example}
\label{example1}
Consider $L = \{a^n b^n \mid n \geq 0\} \in \DCM(1,1)$. Then define a transducer
$A$ with one counter that on input $a$ outputs $a$, on input $b$ outputs $b$, and on
inputs $c$ and $d$ outputs $\lambda$, while verifying that all $c$'s occur
before any $d$'s and that they have the same number of occurrences. Then $A^{-1}(L) = \{w \mid w \mbox{~consists of~} a^nb^n \mbox{~shuffled with~} c^m d^m, n,m \geq 0\}$. Thus, $A^{-1}$ can ``shuffle in'' words with the same number of $c$'s and $d$'s. Alternatively, the same language could be obtained
from $\{\lambda\}$ using the inverse of a transducer $A$ with two counters that checks that the number of $a$'s is the same as the number of $b$'s and
that all $a$'s occur before any $b$'s, and similarly with $c$'s and $d$'s.
\end{example}

In the same way that we attached reversal-bounded counters to transducers, we will briefly consider attaching a single (unrestricted) counter, and also pushdowns.
The following shows that Lemma \ref{invTransductions} and Theorem \ref{DCMCharacterization} do not
generalize for acceptors and transducers with an unrestricted 
counter or with a 1-reversal pushdown.
\begin{theorem}
\begin{enumerate}
\item
There is a language $L$ accepted by a deterministic one counter automaton (i.e., a DFA with one unrestricted counter)
and a deterministic one-counter transducer (i.e. a deterministic one-counter automaton with outputs) $A$ 
such that
$A^{-1}(L)$ is not in $\NPCM$.
\item
There is a language $L$ accepted by a 1-reversal deterministic pushdown automata and a 
deterministic 1-reversal
pushdown transducer (i.e., a 1-reversal deterministic pushdown with outputs) $A$ such that
$A^{-1}(L)$ is not in $\NPCM$.
\end{enumerate}
\end{theorem}
\begin{proof}
For Part 1, let
$L =  \{ a^{i_1} \# a^{i_2} \# a^{i_3} \# \cdots \# a^{i_k} \#  \mid   k \ge 2 \mbox{~is even},   i_1 = 1,  i_{j+1} = i_j  +1 \mbox{~for odd~} j  \}$.
This language can be accepted by a deterministic one-counter automaton.

Construct a deterministic counter transducer $A$ which, on input $w$,
outputs $w$,  and accepts  if the following holds:
\begin{enumerate} 
\item $w$ is of the form  $(a^+\#)^k$ for some even $k \ge 2$.   (The finite-state control can check this.)

\item  In $w$, $i_{j+1} = i_j +1$ for even $j$.  (This needs an unrestricted counter.)
\end{enumerate}
Then  $A^{-1}(L) = \{ a^{i_1} \# a^{i_2} \# a^{i_3} \# \cdots \# a^{i_k} \# \mid  k \ge 2 \mbox{~is even}, i_1 = 1,  i_{j+1} = i_j  +1 \mbox{~for all~} j, 1 \leq j <k \}$.
However, the Parikh map of $A^{-1}(L)$ is not semilinear.  The result follows since the Parikh map of any
$\NPCM$ language is semilinear \cite{Ibarra1978}.

For Part 2, let
$L =  \{ a^{i_1} \# a^{i_3} \# a^{i_5} \# \cdots  \#  a^{i_{2k-1}} \# \$
              a^{i_{2k}} \#   a^{i_{2k-2}} \# \cdots  \# a^{i_2} \# \mid k \ge 1,   i_1 = 1,  i_{j+1} = i_j  +1 \mbox{~for odd~} j \}$.
Then $L$ can be accepted by a 1-reversal deterministic pushdown automaton.

We construct  a  deterministic 1-reversal pushdown transducer $A$ which, on input $w$, outputs $w$,  and accepts  if the following holds:
\begin{enumerate}
\item $w$ is of the form $(a^+\#)^m \$ (a^+\#)^n$ for some even $m, n \ge 1$   (The finite-state control can check this.)
\item  In $w$, $i_{j+1} = i_j +1$ for even $j$.  (This needs a 1-reversal stack).
\end{enumerate}
Then $A^{-1}(L)  =  \{ a^{i_1} \# a^{i_3} \# a^{i_5} \# \cdots  \# a^{i_{2k-1}} \# \$ 
              a^{i_{2k}} \#   a^{i_{2k-2}} \# \cdots  \# a^{i_2} \# \mid  k \ge 1,   i_1 = 1,    i_{j+1} = i_j  +1 \mbox{~for all~} j \}$,
which is not semilinear.
\qed
\end{proof}

However, we have:
\begin{theorem}
If $L$  is in $\DCM$  ($\NCM$) and $A$ is a deterministic (nondeterministic)
transducer with a pushdown and reversal-bounded counters, or
$L$ is in $\DPCM$ ($\NPCM$) and $A$ is a deterministic (nondeterministic)
reversal-bounded counter transducer, then
$A^{-1} (L)$ is in $\DPCM$ ($\NPCM$).
\end{theorem}
\begin{proof}
Similar to the proof of Lemma \ref{invTransductions}.
\end{proof}

\section{Summary of Results}
\label{sec:summary}

This section summarizes insertion closure properties demonstrated in this paper. For one-way machines, all closure properties, both for $\DCM(k,l)$ and $\DCM$ are summarized in Table \ref{tab:summary}.
Also, for two-way machines with one reversal-bounded counter, $2\DCM(1)$, the results are summarized as follows:
\begin{itemize}
\item There exists $L \in \DCM(1,1)$ (one-way), s.t. $\suff^{-1}(L) \notin 2\DCM(1)$ (Theorem \ref{inversesuffix}).
\item There exists $L \in \DCM(1,1)$ (one-way) , $R$ regular, s.t. $RL \notin 2\DCM(1)$ (Corollary \ref{leftreg}).
\item There exists $L \in \DCM(1,1)$ (one-way), s.t. $\outf^{-1}(L) \notin 2\DCM(1)$ (Theorem \ref{inverseoutfix}).
\item There exists $L \in \DCM(1,1)$ (one-way), s.t. $\infx^{-1}(L) \notin 2\DCM(1)$ (Theorem \ref{prop:inverseNotInDCM}).
\item There exists $L \in 2\DCM(1)$, 1 input turn, 1 counter reversal, s.t. $\pref^{-1}(L) \notin 2\DCM(1)$ (Theorem \ref{nonclosure2}).
\item There exists $L \in 2\DCM(1)$, 1 input turn, 1 counter reversal, $R$ regular, s.t. $LR \notin 2\DCM(1)$ (Theorem \ref{nonclosure2}).
\end{itemize}

This resolves every open question summarized above, optimally, in terms of the number of counters, reversals on counters, and reversals on the input tape.
Also, it was shown that the right input end-marker is necessary for
$\DCM$, and that $\DCM$ is closed under inverse deterministic reversal-bounded
multicounter transducers that can define natural insertion operations.








\section*{Acknowledgements} We thank the anonymous reviewers for suggestions improving the presentation of the paper.

\bibliography{undecide_refs}
\bibliographystyle{elsarticle-num}

\end{document}